\documentclass[11pt]{article}

\usepackage[bookmarks=false]{hyperref}

\usepackage[T1]{fontenc}

\def\colorful{1}

\usepackage{amsthm,amsfonts,amsmath,amssymb,epsfig,color,float,graphicx,verbatim}
\usepackage{multirow}

\newif\ifhyper\IfFileExists{hyperref.sty}{\hypertrue}{\hyperfalse}
\hyperfalse
\hypertrue
\ifhyper\usepackage{hyperref}\fi

\usepackage{enumitem}
\usepackage[capitalize]{cleveref} 
\usepackage[makeroom]{cancel}

\makeatletter
\renewcommand{\section}{\@startsection{section}{1}{0pt}{-12pt}{5pt}{\large\bf}}
\renewcommand{\subsection}{\@startsection{subsection}{2}{0pt}{-12pt}{-5pt}{\normalsize\bf}}
\renewcommand{\subsubsection}{\@startsection{subsubsection}{3}{0pt}{-12pt}{-5pt}{\normalsize\bf}}
\makeatother

\usepackage{framed}
\usepackage{nicefrac}

\def\nnewcolor{1}
\ifnum\nnewcolor=1

\fi
\ifnum\nnewcolor=0

\fi

\ifnum\colorful=1
\newcommand{\new}[1]{{\color{black} #1}}

\else
\newcommand{\new}[1]{{#1}}

\fi

\newtheorem{theorem}{Theorem}

\newtheorem{lemma}[theorem]{Lemma}
\newtheorem{proposition}[theorem]{Proposition}

\theoremstyle{definition}

\newcommand{\E}{\mathbb{E}}

\newcommand{\var}{\mathrm{Var}}

\DeclareMathOperator{\cov}{Cov}

\newcommand{\ignore}[1]{}

\newcommand{\eps}{\epsilon}

\newcommand{\Var}{\mathop{\textnormal{Var}}\nolimits}

\newcommand{\Poi}{\mathop{\textnormal{Poi}}\nolimits}

\newcommand{\littlesum}{\mathop{\textstyle \sum}}

\newenvironment{algorithm}[1][\  ] %
{ \rm
\begin{tabbing}
....\=.....\=.....\=.....\=.....\=  \+ \kill
} %
{\end{tabbing} }

\floatstyle{ruled}
\newfloat{Algorithm}{H}{alg}
\newfloat{Subroutine}{H}{sub}

{
\begin{minipage}{1.0\linewidth} \begin{algorithm} %
} { \end{algorithm} \end{minipage} }

\title{Collision-based Testers are Optimal for Uniformity and Closeness}

\author{
Ilias Diakonikolas\thanks{Part of this research was performed when the author was at the University of Edinburgh, 
and while visiting MIT.  Supported in part by a Marie Curie Career Integration grant and an EPSRC grant.}\\
USC\\
\and
Themis Gouleakis\thanks{This material is based upon work supported by the NSF under Grant No. 1420692.}\\
MIT\\
\and
John Peebles\thanks{This material is based upon work supported by the NSF 
Graduate Research Fellowship under Grant No. 1122374, and by the NSF under Grant No. 1065125.}\\
MIT\\
\and
Eric Price\\
UT Austin
}

\begin{document}

\maketitle

\thispagestyle{empty}

\begin{abstract}
We study the fundamental problems of (i) uniformity testing of a discrete distribution, 
and (ii) closeness testing between two discrete distributions with bounded $\ell_2$-norm.  
These problems have been extensively studied in distribution testing
and sample-optimal estimators are known for them~\cite{Paninski:08, CDVV14, VV14, DKN:15}.

In this work, we show that the original collision-based testers proposed for these problems
~\cite{GRdist:00, BFR+:00} are sample-optimal, up to constant factors. 
Previous analyses showed sample complexity upper bounds for these testers that are optimal
as a function of the domain size $n$, but suboptimal by polynomial factors 
in the error parameter $\eps$. Our main contribution is a new tight analysis 
establishing that these collision-based testers are information-theoretically optimal, 
up to constant factors, both in the dependence on $n$ and in the dependence on $\eps$.
\end{abstract}

\thispagestyle{empty}

\section{Introduction}  \label{sec:intro}

\subsection{Background and Our Results} \label{ssec:background}
The generic inference problem in {\em distribution property testing}~\cite{BFR+:00, Batu13} 
(also see, e.g., \cite{Rub12, Canonne15, Goldreich16-notes}) 
is the following: given sample access to one or more unknown distributions,
determine whether they satisfy some global property or are ``far''
from satisfying the property. During the past couple of decades, distribution testing
-- whose roots lie in statistical hypothesis testing~\cite{NeymanP, lehmann2005testing} -- has developed into a mature field.
One of the most fundamental tasks in this field
is deciding whether an unknown discrete distribution is approximately uniform on its domain, known as the problem of 
{\em uniformity testing}. Formally, we want to design an algorithm that, 
given independent samples from a discrete distribution $p$ over $[n]$ and a parameter $\eps>0$, 
distinguishes (with high probability) the case that $p$ is uniform 
from the case that $p$ is $\eps$-far from uniform, i.e., the total variation distance between $p$ and the uniform distribution
over $[n]$ is at least $\eps$. 

Uniformity testing was the very first problem considered in this line of work:
Goldreich and Ron~\cite{GRdist:00}, motivated by the question of testing the expansion of graphs,
proposed a simple and natural uniformity tester that relies on the {\em collision probability} 
of the unknown distribution. The collision probability of a discrete distribution $p$ is the probability that two
samples drawn according to $p$ are equal. The key intuition here is that the uniform distribution
has the minimum collision probability among all distributions on the same domain, and that any distribution
that is $\eps$-far from uniform has noticeably larger collision probability. Formalizing this intuition, 
Goldreich and Ron~\cite{GRdist:00} showed that the collision-based uniformity tester succeeds
after drawing $O(n^{1/2}/\eps^4)$ samples from the unknown distribution. An information-theoretic lower
bound of $\Omega(n^{1/2})$ on the number of samples required by any uniformity tester follows 
from a simple birthday-paradox argument~\cite{GRdist:00, BFFKRW:01}, 
even for constant values of the parameter $\eps$. In subsequent work, 
Paninski~\cite{Paninski:08} showed an information-theoretic lower bound of $\Omega(n^{1/2}/\eps^2)$,
and also provided a matching upper bound of $O(n^{1/2}/\eps^2)$ that holds under the assumption that 
$\eps = \Omega(n^{-1/4})$\footnote{The uniformity tester of~\cite{Paninski:08} relies 
on the number of {\em unique elements}, i.e., the elements
that appear in the sample set exactly once. Such a tester is only meaningful in the regime that 
the total number of samples is smaller than the domain size.}. This lower bound assumption
on $\eps$ is not inherent: As shown in a number of recent works~\cite{VV14, DKN:15} (see also~\cite{Orlitsky:colt12, CDVV14}), 
a variant of Pearson's $\chi^2$-tester can test uniformity with  $O(n^{1/2}/\eps^2)$ samples for all 
values of $n, \eps>0$. The ``chi-squared type'' testers of~\cite{CDVV14, VV14} 
are simple, but are also arguably slightly less natural than the original collision-based uniformity tester~\cite{GRdist:00}. 

Perhaps surprisingly, prior to this work, the sample complexity of the collision uniformity tester 
was not fully understood. In particular, it was not known whether the sample upper bound 
of $O(n^{1/2}/\eps^4)$ -- established in~\cite{GRdist:00} --  is tight for this tester, 
or there exists an improved analysis that can give a better upper bound.
As our first main contribution (Theorem~\ref{thm:collisions-unif}), 
we provide a new analysis of  the collision uniformity tester establishing a tight 
$O(n^{1/2}/\eps^2)$ upper bound on its sample complexity. That is, we show 
that the originally proposed uniformity tester is in fact sample-optimal, up to constant factors.

A related testing problem of central importance in the field is the following: 
Given samples from two unknown distributions $p, q$ over $[n]$
with the promise that $\max\{ \|p\|_2^2, \|q\|_2^2\} \leq b$, distinguish between the cases that $\|p-q\|_2 \leq \eps/2$ 
and $\|p-q\|_2 \ge \eps.$ That is, we want to test the closeness between two unknown distributions with 
small $\ell_2$-norm. (We remark here that the assumption that both $p$ and $q$ have small $\ell_2$-norm is critical in this context.)
The seminal work of Batu {\em et al.}~\cite{BFR+:00} gave a collision-based tester for this problem
that uses $O(b^2/\eps^4 + b^{1/2}/\eps^2)$ samples. Subsequent work by Chan, Diakonikolas, Valiant, and Valiant~\cite{CDVV14}
gave a different ``chi-squared type'' tester that uses $O(b^{1/2}/\eps^2)$; this sample bound was shown~\cite{CDVV14, VV14} 
to be optimal, up to constant factors. 

Similarly to the case of uniformity testing, prior to this work, it was not known whether the analysis of the 
collision-based closeness tester in~\cite{BFR+:00} is tight. As our second contribution, 
we show (Theorem~\ref{thm:collisions-closeness}) that (essentially) 
the collision-based tester of~\cite{BFR+:00} succeeds with $O(b^{1/2}/\eps^2)$ samples,  
i.e., it is sample-optimal, up to constants, for the corresponding problem.

\paragraph{Remark.}
Uniformity testing has been a useful algorithmic primitive for 
several other distribution testing problems as well~\cite{BFFKRW:01, DDSVV13, DKN:15, DKN:15:FOCS, 
CDGR16, Goldreich16}. Notably, Goldreich~\cite{Goldreich16} 
recently showed that the more general problem of testing
the identity of any explicitly given distribution can be reduced to uniformity testing 
with only a constant factor loss in sample complexity. 

The problem of $\ell_2$ closeness testing for distributions with small $\ell_2$ norm 
has been identified as an important algorithmic primitive 
since the original work of Batu {\em et al.}~\cite{BFR+:00} who exploited it to obtain the first $\ell_1$ closeness tester. 
Recently, Diakonikolas and Kane~\cite{DK16} gave a collection of reductions from various distribution testing problems to the above $\ell_2$ closeness 
testing problem. The approach of~\cite{DK16} shows that one can obtain sample-optimal testers for a range of different properties of distributions
by applying an optimal tester for the above problem as a black-box.



\subsection{Overview of Analysis} \label{sec:overview}
We now provide a brief summary of previous analyses and a comparison with our work.
The canonical way to construct and analyze distribution property testers roughly works as follows:
Given $m$ independent samples $s_1, \ldots, s_m$ from our distribution(s), we consider an appropriate
random variable (statistic) $F(s_1, \ldots, s_m)$. If $F(s_1, \ldots, s_m)$ exceeds an appropriately defined threshold $T$,
our tester rejects; otherwise, it accepts. The canonical analysis proceeds by bounding 
the expectation and variance of $F$ for the case that the distribution(s) satisfy the property (completeness),
and the case they are $\eps$-far from satisfying the property (soundness), followed by an application 
of Chebyshev's inequality.

The main difficulty is choosing the statistic $F$ appropriately so that the expectations for the completeness
and soundness cases are sufficiently separated after a small number of samples, 
and at the same time the variance of the statistic is not ``too large''. Typically, the challenging step
in the analysis is bounding from above the variance of $F$ in the soundness case. 
Our analysis follows this standard framework.
Roughly speaking, for both problems we consider, 
we provide a tighter analysis of the variance of the corresponding estimators,
that in turn leads to the optimal sample complexity upper bound.

More specifically, for the case of uniformity testing, the argument of~\cite{GRdist:00} proceeds by 
showing that the collision tester yields a  $(1+\gamma)$-multiplicative approximation of the $\ell_2^2$-norm
of the unknown distribution with $O(n^{1/2}/\gamma^2)$ samples. Setting $\gamma = \eps^2$ gives a uniformity 
testing under the $\ell_1$ distance that uses  $O(n^{1/2}/\eps^4)$ samples. We note that the quadratic 
dependence on $1/\gamma$ in the multiplicative approximation of the $\ell_2$ norm is tight in general.
(For an easy example, consider the case that our distribution is either uniform over two elements, 
or assigns probability mass $1/2-\gamma, 1/2+\gamma$ to the elements.) 
Roughly speaking, we show that we can do better when the $\ell_2$ norm of the distribution in question is small.
More specifically, the collision uniformity tester can distinguish between the case that $\|p\|_2^2 \leq (1+\gamma/2)/n$ and 
$\|p\|_2^2 \geq (1+\gamma)/n$ with $O(n^{1/2}/\gamma)$ samples. This immediately yields the desired $\ell_1$ guarantee.

For the closeness testing problem (under our bounded $\ell_2$ norm assumption), 
Batu {\em et al.} \cite{BFR+:00} construct a statistic whose expectation is proportional to the square of the $\ell_2$ distance between the two 
distributions $p$ and $q$. This statistic has three terms whose expectations are  proportional to $\|p\|_2^2$, $\|q\|_2^2$, and $2 p \cdot q$ respectively. 
Specifically, the first term is obtained by considering the number of self-collisions of a set of samples from $p$. 
Similarly, the second term is proportional to the number of self-collisions of a set of samples from $q$. 
The third term is obtained by considering the number of ``cross-collisions'' between some samples from $p$ and $q$. 
In order to simplify the analysis, \cite{BFR+:00} uses a separate set of fresh samples for the cross-collisions term. 
This set is independent of the set of samples used for the two self-collisions terms. 
While this choice makes the analysis cleaner, it ends up increasing the variance of the estimator too much 
leading to a sub-optimal sample upper bound. We show that by reusing samples to calculate 
the number of cross-collisions, one achieves sufficiently good variance to get optimal sample 
complexity. This comes at the cost of a more complicated analysis involving a very careful calculation of the variance.







\subsection{Notation}
We write $[n]$ to denote the set $\{1, \ldots, n\}$.
We consider discrete distributions over $[n]$, which are functions
$p: [n] \rightarrow [0,1]$ such that $\sum_{i=1}^n p_i =1.$
We use the notation $p_i$ to denote the probability of element
$i$ in distribution $p$. We will denote by $U_n$ the uniform distribution over $[n]$.

For $r \ge 1$, the $\ell_r$--norm of a distribution is identified with the $\ell_r$--norm of the corresponding vector, i.e.,
$\|p\|_r = \left(\sum_{i=1}^n |p_i|^r\right)^{1/r}$. The $\ell_1$ (resp. $\ell_2$) distance between distributions
$p$ and $q$ is defined as the  the $\ell_1$ (resp. $\ell_2$) norm of the vector of their difference, 
i.e., $\|p-q\|_1 = \sum_{i=1}^n |p_i -q_i|$ and $\|p-q\|_2 = \sqrt{\sum_{i=1}^n (p_i-q_i)^2}$.





\section{Testing Uniformity via Collisions} \label{sec:uniform}

In this section, we show that the natural collision uniformity tester 
proposed in~\cite{GRdist:00} is sample-optimal up to constant factors.
More specifically, we are given $m$ samples from a probability distribution $p$ over $[n]$, and 
we wish to distinguish (with high constant probability) between the cases that $p$ is uniform versus $\eps$-far from 
uniform in $\ell_1$-distance. The main result of this section is that the collision-based uniformity tester succeeds in this task
with $m = O(n^{1/2}/\eps^2)$ samples. 

In fact, we prove the following stronger $\ell_2$-guarantee for the collisions tester:
With $m = O(n^{1/2}/\eps^2)$ samples, it distinguishes between the cases that $\|p - U_n\|_2^2 \leq \eps^2/(2n)$ (completeness)
versus $\|p-U_n\|_2^2 \geq \eps^2/n$ (soundness). The desired $\ell_1$ guarantee follows from this $\ell_2$ guarantee 
by an application of the Cauchy-Schwarz inequality in the soundness case.

\medskip

Formally, we analyze the following tester:

\medskip
\fbox{\parbox{6in}{

{\bf Algorithm} \textsc{Test-Uniformity-Collisions$(p, n, \eps)$} \\
Input: sample access to a distribution $p$ over $[n]$, and $\eps>0$.\\
Output: ``YES'' if $\|p - U_n\|_2^2 \leq \eps^2/(2n)$; ``NO'' if $\|p-U_n\|_2^2 \ge \eps^2/n.$
\begin{enumerate}
  \item Draw $m$ iid samples from $p$.
  \item Let $\sigma_{ij}$ be an indicator variable which is $1$ if samples $i$ and $j$ are the same and $0$ otherwise.
  \item Define the random variable $s = \sum_{i<j} \sigma_{ij}$ and the threshold $t=\binom{m}{2} \cdot \frac{1+3\eps^2/4}{n}$
  \item If $s \ge t$ return ``NO''; otherwise, return ``YES''.
\end{enumerate}
}}

\bigskip

The following theorem characterizes the performance of the above estimator:

\begin{theorem} \label{thm:collisions-unif}
The above estimator, when given $m$ samples drawn from a distribution $p$ over $[n]$ will, with probability at least $3/4$,
distinguish the case that $\|p - U_n\|_2^2 \leq \eps^2/(2n)$ from the case 
that $||p-U_n||_2^2 \ge \eps^2/n$ provided that  $m  \ge  3200 n^{1/2}/\eps^2$.
\end{theorem}


The rest of this section is devoted to the proof of Theorem~\ref{thm:collisions-unif}.
Note that the condition of the theorem is equivalent to testing whether 
$\|p\|_2^2 \leq   \frac{1+\eps^2/2}{n}$ versus $\|p\|_2^2 \geq \frac{1+\eps^2}{n}$. 
Our tester takes $m=\frac{3200n^{1/2}}{\epsilon^2}$ samples from $p$ 
and distinguishes between the two cases with probability at least $3/4$.



\subsection{Analysis of \textsc{Test-Uniformity-Collisions}}

The analysis proceeds by bounding the expectation and variance of the estimator
for the completeness and soundness cases, and applying Chebyshev's inequality.
The novelty here is a tight analysis of the variance which leads to the optimal sample  bound.

We start by recalling the following simple closed formula for the expected value:
\begin{lemma} \label{lem:expectation-collisions} We have that
$
\E[s] = \binom{m}{2} \|p\|_2^2 \;.
$
\end{lemma}

\begin{proof}
For any $i,j$, the probability that samples $i$ and $j$ are equal is $\|p\|_2^2$. 
By this and linearity of expectation, we get
\[
\E[s] = \E\left[\sum_{ij} \sigma_{ij}\right] = \sum_{ij} \E[\sigma_{ij}] = \sum_{ij} \|p\|_2^2 = \binom{m}{2} \|p\|_2^2.\]
\end{proof}

Thus, we see that in the completeness case the expected value is at most $\binom{m}{2} \cdot \frac{1+\eps^2/2}{n}$. 
In the soundness case,
the expected value is at least $\binom{m}{2} \cdot \frac{1+\epsilon^2}{n}$. 
This motivates our choice of the threshold $t$ halfway between these expected values.

In order to argue that the statistic will be close to its expected value, we bound its variance from above and use Chebyshev's inequality. 
We bound the variance in two steps. First, we obtain the following bound:
\begin{lemma}\label{lem:variance} We have that 
$\var[s] \leq m^2 \cdot \|p\|_2^2 + m^3 \cdot ( \|p\|_3^3 -\|p\|_2^4).$
\end{lemma}

\begin{proof} The lemma follows from the following chain of (in-)equalities:
\begin{align*}
\var[s] &= \E[s^2] - \E[s]^2 \\
&= \E\left[ \sum_{i<j} \sum_{k < \ell} \sigma_{ij} \sigma_{k\ell} \right] - \binom{m}{2}^2 \|p\|_2^4\\
&= \E\left[ \sum_{\substack{i<j; \; k<\ell \\ \text{all distinct}}} \sigma_{ij} \sigma_{k\ell} + 2\sum_{i<j<\ell} \sigma_{ij}\sigma_{j\ell} + 2\sum_{\substack{i,k<j\\i \neq k}} \sigma_{ij}\sigma_{kj} + \sum_{i<j} \sigma_{ij}^2 \right] - \binom{m}{2}^2 \|p\|_2^4 \\
&= \binom{m}{2} \binom{m-2}{2} \|p\|_2^4 + 2 \cdot \binom{m}{3} \|p\|_3^3 + 4 \cdot \binom{m}{3} \|p\|_3^3 + \binom{m}{2} \|p\|_2^2 - \binom{m}{2}^2 \|p\|_2^4 \\
&= \binom{m}{2} \cdot (\|p\|_2^2 - \|p\|_2^4) + m(m-1)(m-2) \cdot ( \|p\|_3^3 -\|p\|_2^4) \\
&\leq m^2 \cdot \|p\|_2^2 + m^3 \cdot ( \|p\|_3^3 -\|p\|_2^4).
\end{align*}
\end{proof}

\paragraph{Remark.}
We note that the upper bound of the previous lemma is tight, up to constant factors. 
The $-m^3\|p\|_2^4$ term is critical for getting the optimal dependence on $\eps$ in the sample bound. 

\medskip

Continuing the analysis, we now derive an upper bound on the number 
of samples that suffices for the tester to have the desired success probability of $3/4$.

\begin{lemma}\label{lem:samples}
Let $\alpha$ satisfy $\|p\|_2^2 = \frac{1+\alpha}{n}$ and 
 $\sigma$ be the standard deviation of $s$. The number of samples required by \textsc{Test-Uniformity-Collisions} is at most
\[
m \leq \sqrt{\frac{5 \sigma n}{|\alpha-3\epsilon^2/4|}} \;,
\]
in order to get error probability at most $1/4$.
\end{lemma}

\begin{proof}
By Chebyshev's inequality, we have that
\[
\Pr\left[\;\left|s-\binom{m}{2}\|p\|_2^2 \right| \geq k \sigma\right] \leq \frac{1}{k^2}
\]
where $\sigma \triangleq \sqrt{\var[s]}$.

We want $s$ to be closer to its expected value than the threshold is to its expected 
value because when this occurs, the tester outputs the right answer. 
Furthermore, to achieve our desired probability of error of at most $1/4$, 
we want this to happen with probability at least $3/4$. 
So, we set $k=2$, and then we want
\[
k \sigma \leq |\E[s]-t| =\left| \binom{m}{2} |\|p\|_2^2 - (1+3\epsilon^2/4)/n\right| = \binom{m}{2} |\alpha - 3\epsilon^2/4| / n
\]
It suffices for the number of samples $m$ to satisfy the slightly stronger condition that
\[
\sigma \leq m^2 \cdot \frac{|\alpha-3\epsilon^2/4|}{5n}.
\]
So, it suffices to have
\[
m \geq \sqrt{\frac{5 \sigma n}{|\alpha-3\epsilon^2/4|}}.
\]
We might as well take the smallest number of samples $m$ for which the tester works, which implies the desired inequality.
\end{proof}

We are now ready to show an upper bound on the number of samples in the completeness case, 
i.e., when $p$ is the uniform distribution.
\begin{lemma}
In the completeness case,  the required number of samples is at most
\[
m \leq \frac{6n^{1/2}}{\epsilon^2} \;,
\]
in order to get error probability $1/4$.
\end{lemma}

\begin{proof}
It is easy to see that $\|p\|_2^2=1/n$ and $\|p\|_3^3=\|p\|_2^4=1/n^2$. 
Thus, by Lemma \ref{lem:variance}, $\sigma \leq m/n^{1/2}$. 
Also, we know $\alpha=0$ when $p$ is uniform. 
Substituting these two facts into Lemma \ref{lem:samples} and solving for $m$ gives
\[
m \leq \frac{6n^{1/2}}{\epsilon^2}.
\]
\end{proof}

We now turn to the soundness case, 
where $p$ is far from uniform. 
By Lemma \ref{lem:samples}, it suffices to bound from above the variance $\sigma^2$. 
We proceed by a case analysis based on 
whether the term $m^2 \|p\|_2^2$ or $m^3 (\|p\|_3^3 - \|p\|_2^4)$ 
contributes more to the variance.

\subsubsection{Case when $m^2\|p\|_2^2$ is Larger}

\begin{lemma}
Consider the soundness case, 
where $\|p\|_2^2=(1+\alpha)/n$ for $\alpha \geq \epsilon^2$. 
If $m^2\|p\|_2^2$ contributes more to the variance, 
i.e., if $m^2\|p\|_2^2 \geq m^3 (\|p\|_3^3 - \|p\|_2^4)$, 
then the required number of samples is at most
\[
m \leq \frac{48n^{1/2}}{\epsilon^2}
\]
in order to get error probability $1/4$.
\end{lemma}

\begin{proof}
Suppose that $m^2\|p\|_2^2 \geq m^3 (\|p\|_3^3 - \|p\|_2^4)$. 
Then $\sigma^2 \leq 2 m^2 \|p\|_2^2 = 2 m^2 (1+\alpha)/n$. 
Substituting this into Lemma \ref{lem:samples} and solving for $m$ gives that the necessary number of samples is at most
\[
m \leq 8 n^{1/2} \cdot \frac{\sqrt{1+\alpha}}{(\alpha-3\epsilon^2/4)}.
\]
Using calculus to maximize this expression by varying $\alpha$, 
one gets that $\alpha=\epsilon^2$ maximizes the expression. Thus,
\[
m \leq 32 n^{1/2} \cdot \frac{\sqrt{1+\epsilon^2}}{\epsilon^2} \leq 32 n^{1/2} \cdot \frac{\sqrt{2}}{\epsilon^2} 
< \frac{48 n^{1/2}}{\epsilon^2}.
\]
\end{proof}

\subsubsection{Case when $m^3(\|p\|_3^3 - \|p\|_2^4)$ is Larger}

\begin{lemma}
Consider the soundness case, 
where $\|p\|_2^2=(1+\alpha)/n$ for $\alpha \geq \epsilon^2$. 
If $m^3 (\|p\|_3^3 - \|p\|_2^4)$ contributes more to the variance, 
i.e., if $m^3 (\|p\|_3^3 - \|p\|_2^4) \geq m^2\|p\|_2^2$, 
then the required number of samples is at most
\[
m \leq \frac{3200n^{1/2}}{\epsilon^2}
\]
in order to get error probability $\leq 1/4$.
\end{lemma}

\begin{proof}
Suppose that $m^3 (\|p\|_3^3 - \|p\|_2^4) \geq m^2\|p\|_2^2$. 
Then $\sigma^2 \leq 2 m^3 (\|p\|_3^3 - \|p\|_2^4)$. 
Substituting this into Lemma \ref{lem:samples} and solving for $m$ 
gives that the necessary number of samples is at most
\[
m \leq 50 n^2 \cdot \frac{\|p\|_3^3 - \|p\|_2^4}{(\alpha - 3\epsilon^2/4)^2}.
\]
Let us parameterize $p$ as $p_i = 1/n + a_i$ for some vector $a$. 
Then we have $\|a\|_2^2 = \alpha/n$, and we can write
\begin{align*}
m &\leq 50 n^2 \cdot \frac{\|p\|_3^3 - \|p\|_2^4}{(\alpha - 3\epsilon^2/4)^2} \leq 50n^2 \cdot \frac{\|p\|_3^3 - \|p\|_2^4}{({\alpha/4})^2} & \text{(since $\epsilon^2 \leq \alpha$)}\\
&\leq 50n^2 \cdot \frac{\|p\|_3^3 - \frac{1}{n^2}}{({\alpha/4})^2} 
= 50n^2 \cdot \frac{\left[ \sum_{i=1}^n (1/n + a_i)^3 \right] - \frac{1}{n^2}}{({\alpha/4})^2} \\
&= 50n^2 \cdot \frac{\left[ \frac{1}{n^2} + \frac{3}{n^2} \sum_{i=1}^n a_i + 
\frac{3}{n} \sum_{i=1}^n a_i^2 + \sum_{i=1}^n a_i^3 \right] - \frac{1}{n^2}}{({\alpha/4})^2} \\
&= 50n^2 \cdot \frac{ \frac{3}{n^2} \sum_{i=1}^n a_i + \frac{3}{n} \sum_{i=1}^n a_i^2 + \sum_{i=1}^n a_i^3}{({\alpha/4})^2} \\
&= 50n^2 \cdot \frac{ \frac{3}{n} \sum_{i=1}^n a_i^2 + \sum_{i=1}^n a_i^3}{({\alpha/4})^2} & \text{(since $\littlesum_{i=1}^n a_i = 0$)} \\
&\leq 50n^2 \cdot \frac{ \frac{3}{n} \|a\|_2^2 + \|a\|_3^3}{({\alpha/4})^2} \leq 50n^2 \cdot \frac{ \frac{3}{n} \|a\|_2^2 + \|a\|_2^3}{({\alpha/4})^2} \\
&= 50n^2 \cdot \frac{ \frac{3}{n} (\alpha/n) + (\alpha/n)^{3/2}}{({\alpha/4})^2} = \frac{2400}{\alpha} + \frac{{800}n^{1/2}}{\sqrt{\alpha}}\\
&\leq \frac{2400}{\epsilon^2} + \frac{800n^{1/2}}{\sqrt{\epsilon^2}} & \text{(since $\epsilon^2 \leq \alpha$)}\\
&\leq \frac{3200n^{1/2}}{\epsilon^2} \;.
\end{align*}
\end{proof}

Note that, as mentioned earlier, if we had ignored the $-\|p\|_2^4$ term, 
we would have had an $\Omega(1/\epsilon^4)$ term in our bound, 
which would have given us the wrong dependence on $\epsilon$. 

Theorem~\ref{thm:collisions-unif} now follows as an immediate consequence of these last three lemmas.

\paragraph{Remark.}
It is worth noting that the collisions statistic analyzed in this section is very similar to the chi-squared-like uniformity tester 
in \cite{DKN:15} -- itself a simplification of similar testers in~\cite{CDVV14, VV14} -- which also achieves 
the optimal sample complexity of $O(n^{1/2} / \epsilon^2)$. 
Specifically, if $X_i$ denotes the number of times we see the $i$-th domain element in the sample, 
the \cite{DKN:15} statistic is $\sum_i (X_i - m/n)^2 - X_i = \sum_{i<j} \sigma_{ij}  -2\frac{m}{n} \sum_i X_i +\frac{ m^2}{n}$. 
We note that the \cite{DKN:15} analysis uses Poissonization; i.e., instead of drawing $m$ samples from the distribution, we draw $\Poi(m)$ samples. 
Without Poissonization, the aforementioned statistic simplifies to $s - \frac{ m^2}{n}$, where $s$ is the collisions statistic. 
While the non-Poissonized versions of the two testers are equivalent, the Poissonized versions are not. 
Specifically, the Poissonized version of the  \cite{DKN:15} uniformity tester has sufficiently good variance to yield the sample-optimal bound. 
On the other hand, the Poissonized version of the collisions statistic 
does not have good variance: Specifically, its variance 
does not have the $-\|p\|_2^4$ term which -- as noted earlier -- is necessary to get the optimal $\epsilon$ dependence.

\section{Testing  Closeness via Collisions} \label{sec:equiv}

Given samples from two unknown distributions $p, q$ over $[n]$
with the promise that $\max\{ \|p\|_2^2, \|q\|_2^2\} \leq b$, we want to distinguish between the cases that $\|p-q\|_2 \leq \eps/2$ 
versus $\|p-q\|_2 \ge \eps.$ We show that a natural collisions-based tester succeeds in this task with $O(b^{1/2}/\eps^2)$ samples.
The estimator we analyze is a slight variant of the $\ell_2$ tester in~\cite{BFR+:00}, described in pseudocode below.

We define the number of self-collisions in a sequence of samples from a distribution as $\sum_{i<j} \sigma_{ij}$, 
where $\sigma_{ij}$ is the indicator variable denoting whether samples $i$ and $j$ are the same. 
Similarly, we define the number of cross-collisions between two sequences of samples as $\sum_{i,j} \ell_{ij}$, 
where $\ell_{ij}$ is the indicator variable denoting 
whether sample $i$ from the first sequence is the same as sample $j$ from the second sequence.

\bigskip

\fbox{\parbox{6in}{

{\bf Algorithm} \textsc{Test-Closeness-Collisions}$(p, q, n, b, \eps)$ \\
Input: sample access to distribution $p, q$ over $[n]$, $\eps, b>0$.\\
Output: ``YES'' if $\|p-q\|_2 \leq \eps/2$; ``NO'' if $\|p-q\|_2 \ge \eps.$
\begin{enumerate}
  \item Draw two multisets $S_p, S_q$ of $m$ iid samples from $p, q$. 
  Let $C_1$ denote the number of self-collisions of $S_p$, 
  $C_2$ denote the number of self-collisions of $S_q$, 
  and $C_3$ denote the number of cross-collisions between $S_p$ and $S_q$.
  \item Define the random variable $Z = C_1 + C_2 - \frac{m-1}{m} \cdot C_3$ 
  and the threshold $t=\new{\binom{m}{2}} \epsilon^2/2$. 
  \item If $Z \ge t$ return ``NO''; otherwise, return ``YES''.
\end{enumerate}
}}

\bigskip

The following theorem characterizes the performance of the above estimator:

\begin{theorem}\label{thm:collisions-closeness}
There exists an absolute constant $c$ such that the above estimator, 
when given $m$ samples drawn from each of two distributions, $p, q$ over $[n]$ will, with probability at least $3/4$, 
distinguish the case $||p-q||_2 \le \eps/2$ from the case 
that $||p-q||_2 \ge \eps$ provided that $m \ge c \cdot \frac{b^{1/2}}{\eps^2},$ where $b$ is an upper bound on $||p||_2^2, ||q||_2^2$.
\end{theorem}

\subsection{Analysis of \textsc{Test-Closeness-Collisions}}
Let $X_i, Y_i$ be the number of times we see the element $i$
in each set of samples $S_p$ and $S_q$, respectively. 
The above random variables are distributed as follows:  
$X_i\sim Bin(m,p_i), Y_i\sim Bin(m,q_i)$.   
Note that the statistic $Z$ can be written as
\[
Z=\frac{m-1}{\new{2m}}\sum_{i=1}^n \left[(X_i-Y_i)^2-X_i-Y_i \right]+
\frac{1}{\new{2m}}\sum_{i=1}^n \left[ X_i(X_i-1)+Y_i(Y_i-1) \right] = 
\frac{m-1}{\new{2m}}A+\frac{1}{\new{2m}}B \;,
\]
where $A=\sum_{i=1}^n \left[ (X_i-Y_i)^2-X_i-Y_i \right]$ 
and $B=\sum_{i=1}^n \left[ X_i(X_i-1)+Y_i(Y_i-1) \right]$.
Note that 
$$\Var[Z] \leq 4\cdot \max\left\{\frac{(m-1)^2}{\new{4m^2}}\Var[A],\frac{1}{\new{4m^2}}\Var[B]\right\} \;.$$
Note that $B$ essentially corresponds to the number of collisions within two disjoint sets of samples,
hence we already have an upper bound on its variance. The bulk of the analysis
goes into bounding from above
the variance of $A=\sum_{i=1}^n A_i=\sum_{i=1}^n \left[(X_i-Y_i)^2-X_i-Y_i \right]$. 

\paragraph{Remark.}
The $\ell_2$ collision-based tester we analyze here is closely related to the $\ell_2$-tester of \cite{CDVV14}. 
Specifically, the $A$ term in the expression for $Z$ has the same formula as the $\ell_2$-tester of \cite{CDVV14}. 
However, a key difference is that the statistic of \cite{CDVV14} is Poissonized, which is crucial for its analysis.

\medskip

We now proceed to analyze the collision-based closeness tester.
We start with a simple formula for its expectation:

\new{
\begin{lemma}\label{closeness_expectation}
For the expectation of the statistic $Z$ in the closeness tester, we have: 
\begin{equation} \label{exp}
\mathbb{E}[Z]=\binom{m}{2} \Vert p-q \Vert_2^2 \;.
\end{equation}
\end{lemma}
\begin{proof}
Viewing $p$ and $q$ as vectors, we have
\[
\mathbb{E}[Z] = \mathbb{E}[C_1 + C_2 - \frac{m-1}{m} \cdot C_3] = \binom{m}{2} (p \cdot p) + \binom{m}{2} (q \cdot q) - \frac{m-1}{m} \cdot m^2 (p \cdot q) = \binom{m}{2} \Vert p-q \Vert_2^2.
\]
\end{proof}}

For the variance, we show the following upper bound:

\begin{lemma}\label{closeness_variance}
For the variance of the statistic $Z$ in the closeness tester, we have:
\[
\var[Z]\leq \new{116}m^2b+\new{16}m^3 \Vert p-q \Vert_4^2 b^{1/2} \;.
\]
\end{lemma}

To prove this lemma, we will use the following proposition, whose proof is deferred to the following subsection.

\begin{proposition} \label{prop:var-a}
We have that $\var[A] \leq 100m^2b+8  m^3\sum_i (p_i-q_i) (p_i^2-q_i^2)$.
\end{proposition}

\begin{proof}[Proof of Lemma~\ref{closeness_variance}]
Recall that by Lemma \ref{lem:variance} we have 
\[
\var[B]\leq \new{4}m^2 (\Vert p \Vert_2^2 + \Vert q   \Vert_2^2)+\new{4}m^3 (\Vert p \Vert_3^3-\Vert p \Vert_2^4+\Vert q \Vert_3^3-\Vert q \Vert_2^4)\;.
\]
Combined with Proposition~\ref{prop:var-a}, we obtain:
\begin{align*}
\var[Z]&\leq 4 \cdot \max\left\{\frac{(m-1)^2}{4m^2}\var[A], \frac{1}{4m^2}\var[B]\right\}\\
&\leq \max \Large\{ 100m^2b+8  m^3 \sum_i (p_i-q_i) (p_i^2-q_i^2), \\
&\qquad\qquad \new{4}(\Vert p \Vert_2^2 + \Vert q   \Vert_2^2)+\new{4}m (\Vert p \Vert_3^3-\Vert p \Vert_2^4+\Vert q \Vert_3^3-\Vert q \Vert_2^4) \Large\} \;.
\end{align*}
\new{The second term in the max statement is at most $16mb$. Thus, we have}
\begin{align*}
\var[Z] 
&\leq \new{116 (m-1)^2 b+8m(m-1)^2\sum_i (p_i-q_i) (p_i^2-q_i^2)}\\
&\leq 116m^2b+8m^3\sum_i (p_i-q_i)^2 (p_i+q_i) \\
&\leq 116m^2b+8m^3\sqrt{\sum_i (p_i-q_i)^4 \sum_i (p_i+q_i)^2} & \text{(by the Cauchy-Schwarz inequality)} \\
&\leq 116m^2b+16m^3 \Vert p-q \Vert_4^2 b^{1/2} & \text{(since $\sum_i (p_i+q_i)^2\leq 4b$)} \;.
\end{align*}
\end{proof}

\subsection{Proof of Theorem \ref{thm:collisions-closeness}}
\new{
By Lemma~\ref{closeness_variance}, we have that
\[
\var[Z]\leq 116m^2b+16m^3 \Vert p-q \Vert_4^2 b^{1/2} \leq 116m^2b+16m^3 \Vert p-q \Vert_2^2 b^{1/2}.
\]
We wish to show we can distinguish the completeness case (i.e., $\Vert p-q\Vert_2 \leq \epsilon /2$) from the soundness case (i.e., $\Vert p-q\Vert_2 \geq \epsilon$). Set $\alpha = \Vert p-q\Vert_2^2$. Then we are promised that either $\alpha \geq \epsilon^2$ or $\alpha \leq \epsilon^2/4$. Recall we chose $t=\new{\frac{\binom{m}{2}\epsilon^2}{2}}$ and that Lemma~\ref{closeness_expectation} says that $\mathbb{E}[Z] = \binom{m}{2} \alpha$.

Since
\[
\mathbb{E}[Z|\text{completeness case}] \leq t \leq \mathbb{E}[Z|\text{soundness case}],
\]
the only way we fail to distinguish the completeness and soundness cases is if $Z$ deviates from its expectation additively by at least
\[
|t-\mathbb{E}[Z]| = \left| \frac{\binom{m}{2}\epsilon^2}{2} - \binom{m}{2}\alpha \right| \geq \binom{m}{2} \max\{\alpha,\epsilon^2\} / 4,
\]
where the last inequality follows by the promise on $\alpha$ in the completeness and soundness cases.\footnote{In the completeness case where $\alpha \leq \varepsilon^2/4$ and $\mathbb{E}[Z]=\binom{m}{2}\alpha$, $Z$ has to deviate by at least $\binom{m}{2}\varepsilon^2/4 \geq \binom{m} \epsilon^2 \alpha$ to cross the threshold $t=\binom{m}{2}\varepsilon^2/2$. In the soundness case where $\alpha \geq \varepsilon^2$, $Z$ has to deviate by at least $\binom{m}{2}\alpha/2 \geq \epsilon / 2$ to cross the threshold $t$.} By Chebyshev's inequality, the probability this happens is at most
\begin{align*}
\Pr[\;|Z - \mathbb{E}[Z]| \geq \binom{m}{2} \max\{\alpha,\epsilon^2\} / 4\;] &\leq \frac{\var[Z]}{[t-\mathbb{E}[Z]]^2} \leq \frac{116m^2b+16m^3 \alpha b^{1/2}}{[\binom{m}{2} \max\{\alpha,\epsilon^2\} / 4]^2} \\
&\leq \frac{32768 \cdot b}{m^2 \epsilon^4} + \frac{4096 \cdot b^{1/2}}{m} \cdot \min\left\{ \frac{1}{\alpha}, \frac{\alpha}{\epsilon^4} \right\} \\
&\leq \frac{32768 \cdot b}{m^2 \epsilon^4} + \frac{4096 \cdot b^{1/2}}{m \epsilon^2},
\end{align*}
where we simplified using the assumption that $m \geq 2$. Thus, if we set $m=O(\frac{b^{1/2}}{\epsilon^2})$, we get a constant probability of error in both cases as desired.%
} \qed

\subsection{Proof of Proposition~\ref{prop:var-a}}
Recall that  $A=\sum_{i=1}^n A_i=\sum_{i=1}^n \left[(X_i-Y_i)^2-X_i-Y_i \right]$,
hence $\var(A)=\sum_{i=1}^n \var(A_i)+\sum_{i\not= j} \cov(A_i,A_j)$.
We proceed to bound from above the individual variances and covariances
via a sequence of elementary but quite tedious calculations.

\subsubsection{Bounding $\var(A_i)$:}
Since 
\[
A_i=(X_i-Y_i)^2-X_i-Y_i=X_i^2+Y_i^2-2X_iY_i-X_i-Y_i \;,
\]
we can write:
\begin{align*}
\var(A_i)
&=\var(X_i^2)+\var(Y_i^2)+4\var(X_iY_i)+\var(X_i)+\var(Y_i)\\
&+2 \cdot [-2\cov(X_i^2,X_iY_i)-\cov(X_i^2,X_i)-2\cov(Y_i^2,X_iY_i)-\cov(Y_i^2,Y_i)\\
&+2\cov(X_iY_i,X_i)+ 2 \cov(X_iY_i,Y_i)] \;.
\end{align*}
We proceed to calculate the individual quantities:
\begin{itemize}
\item[(a)]\begin{align*}
\cov(X_i^2,X_i)&= \sum_{r,s,t\in[m]} \cov([\sigma_r=\sigma_s=i],[\sigma_t=i])\\
&= \sum_{r\in[m]} \cov([\sigma_r=i],[\sigma_r=i]) + 2\sum_{r,s\in[m],\; r\neq s} \cov([\sigma_r=\sigma_s=i],[\sigma_r=i])\\
&=mp_i(1-p_i)+2(m^2-m)(\mathbb{E}[[\sigma_r=\sigma_s=i]\cdot [\sigma_r=i]]-p_i^2p_i)\\
&=mp_i(1-p_i)+2(m^2-m)(p_i^2-p_i^3)\\
&=mp_i(1-p_i)[1+2p_i(m-1)]\\
&=mp_i(1-p_i)[1-2p_i+2p_im]\\
&=mp_i(1-p_i)(1-2p_i)+2m^2p_i^2(1-p_i) \;.
\end{align*}

\item[(b)] \[
\cov(X_i^2,X_iY_i)=\mathbb{E}[X_i^3Y_i]-\mathbb{E}[X_i^2]\cdot \mathbb{E}[X_iY_i]=\cov(X_i^2,X_i)\cdot \mathbb{E}[Y_i]={m^2}p_iq_i(1-p_i)(1-2p_i)+2{m^3}p_i^2q_i(1-p_i) \;.
\]

\item[(c)] \[
\cov(X_i,X_iY_i)=\var(X_i)\cdot \mathbb{E}[Y_i]={m^2}p_i(1-p_i)q_i \;.
\]

\item[(d)]
\begin{align*}
\var(X_i^2)=&\mathbb{E}[X_i^4]-(\mathbb{E}[X_i^2])^2\\
=& mp_i(1-7p_i+7mp_i+12p_i^2-18mp_i^2+6m^2p_i^2-6p_i^3\\
&+11mp_i^3-6m^2p_i^3+m^3p_i^3)-(mp_i-mp_i^2+m^2p_i^2)^2\\
=&mp_i-7mp_i^2+7m^2p_i^2+12mp_i^3-18m^2p_i^3+6m^3p_i^3-6mp_i^4+11m^2p_i^4-6m^3p_i^4+m^4p_i^4 \\
&- (m^2p_i^2+m^2p_i^4+m^4p_i^4-2m^2p_i^3+2m^3p_i^3-2m^3p_i^4)\\
=& mp_i-7mp_i^2+6m^2p_i^2+12mp_i^3-16m^2p_i^3+4m^3p_i^3-6mp_i^4+10m^2p_i^4-4m^3p_i^4 \\
=& mp_i-7mp_i^2+6m^2p_i^2+{12mp_i^3-6mp_i^4}-16m^2p_i^3+4m^3p_i^3+10m^2p_i^4-4m^3p_i^4 \;.
\end{align*}

\item[(e)]
\begin{align*}
\var(X_iY_i)& =\mathbb{E}[X_i^2Y_i^2]-(\mathbb{E}[X_iY_i])^2=\mathbb{E}[X_i^2]\mathbb{E}[Y_i^2]-(\mathbb{E}[X_i]\mathbb{E}[Y_i])^2\\
&= (mp_i-mp_i^2+m^2p_i^2)\cdot (mq_i-mq_i^2+m^2q_i^2)-m^4p_i^2q_i^2\\
&= m^2p_iq_i +m^2p_i^2q_i^2-m^2(p_iq_i^2+p_i^2q_i)+m^3(p_iq_i^2+p_i^2q_i)-2m^3p_i^2q_i^2 \;.
\end{align*}
\end{itemize}
So, we get:
\begin{align*}
\var(A_i) &= mp_i-7mp_i^2+6m^2p_i^2+{12mp_i^3-6mp_i^4}-16m^2p_i^3+4m^3p_i^3+10m^2p_i^4-4m^3p_i^4\\ 
&+mq_i-7mq_i^2+6m^2q_i^2+6mq_i^3-16m^2q_i^3+4m^3q_i^3+10m^2q_i^4-4m^3q_i^4 \\
&+4(m^2(p_iq_i+p_i^2q_i^2-p_iq_i^2-p_i^2q_i)+m^3(p_iq_i^2+p_i^2q_i)-2m^3p_i^2q_i^2)\\
&+mp_i(1-p_i)+mq_i(1-q_i)-{4}( {m^2}p_i(1-p_i)(1-2p_i)+{2m^3}p_i^2(1-p_i))q_i\\
&-{2}(mp_i(1-p_i)(1-2p_i)+{4}m^2p_i^2(1-p_i))-{4}( {m^2}q_i(1-q_i)(1-2q_i)+{2m^3}q_i^2(1-q_i))p_i\\
&-{2}mq_i(1-q_i)(1-2q_i)- {4}m^2q_i^2(1-q_i)+ {4m^2}p_i(1-p_i)q_i+ {4m^2}q_i(1-q_i)p_i\\
&=m[p_i-7p_i^2+ {12p_i^3-6p_i^4}+q_i-7q_i^2+ {12q_i^3-6q_i^4}+p_i-p_i^2+q_i-q_i^2\\
&\;\;- {2}p_i(1-p_i)(1-2p_i)- {2}q_i(1-q_i)(1-2q_i)]\\
&+m^2[ {-4p_iq_i(1-p_i)(1-2p_i)-4p_iq_i(1-q_i)(1-2q_i)+4p_iq_i(2-p_i-q_i)}\\
&\;\;+6p_i^2-16p_i^3+10p_i^4+6q_i^2-16q_i^3+10q_i^4+4p_iq_i(1+p_iq_i-p_i-q_i)\\
&\;\;- {4}p_i^2+ {4}p_i^3- {4}q_i^2+ {4}q_i^3]\\
&+m^3[4p_i^3-4p_i^4+4q_i^3-4q_i^4+4p_iq_i(p_i+q_i)-8p_i^2q_i^2 {-8p_1^2q_i-8p_iq_i^2+8p_1^3q_i+8p_iq_i^3}]\\
&= m [-2 p_i^2 + 8 p_i^3 - 6 p_i^4 - 2 q_i^2 + 8 q_i^3 - 6 q_i^4] \\
&+ m^2 [2(p_i+q_i)^2 - 12 p_i^3 + 10 p_i^4 + 4 p_i^2 q_i - 8 p_i^3 q_i + 4 p_i q_i^2 + 4 p_i^2 q_i^2 - 12 q_i^3 - 8 p_i q_i^3 + 10 q_i^4)] \\
&+ 4m^3 (p_i-q_i)^2 [p_i(1-p_i) + q_i(1-q_i)] \\
&\leq 8m (p_i^3 +q_i^3) + 12 m^2 (p_i+q_i)^2 + 4m^3 (p_i-q_i)^2 \left(p_i+q_i\right)  \\
& \leq 20m^2 (p_i+ q_i)^2 + 4m^3 (p_i-q_i)^2 (p_i+q_i) \;.\\
\end{align*}

\subsection{Bounding the Covariances}
It suffices to show that the covariances of $A_i$ and $A_j$, for $i \neq j$, are appropriately bounded from above. 
Let $i \neq j$. Note that if $\sigma_r$ is the result of sample $k$, we have:
\[
\cov(X_i,X_j) = \sum_{r,u\in[m]} \cov([\sigma_r=i],[\sigma_u=j]) = \sum_{r\in[m]} \cov([\sigma_r=i],[\sigma_r=j]) = -m p_i p_j \;.
\]
Similarly,
\begin{align*}
\cov(X_i^2,X_j) &= \sum_{r,s,t\in[m]} \cov([\sigma_r=\sigma_s=i],[\sigma_t=j]) \\
&= \sum_{r\in[m]} \cov([\sigma_r=i],[\sigma_r=j]) + 2\sum_{r,s\in[m],\; r\neq s} \cov([\sigma_r=\sigma_s=i],[\sigma_r=j]) \\
&= -m p_i p_j -2m(m-1) p_i^2 p_j.
\end{align*}

Similarly,
\begin{align*}
\cov(X_i^2,X_j^2) =& \sum_{r,s,t,u\in[m]} \cov([\sigma_r=\sigma_s=i],[\sigma_t=\sigma_u=j]) \\
=& \; 4\sum_{\text{unique } r,s,u\in[m]} \cov([\sigma_r=\sigma_s=i],[\sigma_r=\sigma_u=j]) \\
&+ 2\sum_{\text{unique } r,s\in[m]} \cov([\sigma_r=\sigma_s=i],[\sigma_r=\sigma_s=j]) \\
&+ 2\sum_{\text{unique } r,s\in[m]} \cov([\sigma_r=\sigma_s=i],[\sigma_r=j]) \\
&+ 2\sum_{\text{unique } r,t\in[m]} \cov([\sigma_r=i],[\sigma_r=\sigma_t=j]) \\
&+ \sum_{r\in[m]} \cov([\sigma_r=i],[\sigma_r=j]) \\
&= -mp_ip_j -2m(m-1)(p_i^2 p_j + p_i p_j^2 + p_i^2 p_j^2) -4m(m-1)(m-2)p_i^2 p_j^2 \\
&= -mp_ip_j -2m(m-1)(p_i^2 p_j + p_i p_j^2) -2m(m-1)(2m-3) p_i^2 p_j^2.
\end{align*}
And,
\begin{align*}
\cov(X_i Y_i,X_j Y_j) &= \E[X_i Y_i X_j Y_j] - \E[X_i Y_i] \E[X_j Y_j] \\
&= \E[X_i X_j] \E[Y_i Y_j] - \E[X_i]\E[Y_i]\E[X_j]\E[Y_j] \\
&= (\cov(X_i,X_j)+\E[X_i]\E[X_j]) \cdot (\cov(Y_i,Y_j)+\E[Y_i]\E[Y_j])-\E[X_i]\E[X_j]\E[Y_i]\E[Y_j] \\
&= (m^2-2m^3)p_i p_j q_i q_j.
\end{align*}
Also,
\begin{align*}
\cov(X_i Y_i,X_j) &= \E[X_i Y_i X_j] - \E[X_i Y_i] \E[X_j] \\
&= \E[X_i X_j] \E[Y_i] - \E[X_i]\E[Y_i]\E[X_j] \\
&= (\cov(X_i,X_j)+\E[X_i]\E[X_j]) \cdot \E[Y_i]-\E[X_i]\E[X_j]\E[Y_i] \\
&= \cov(X_i,X_j) \E[Y_i] \;.
\end{align*}
Similar equations hold if we swap $i$ and $j$ and/or we swap $X$ and $Y$. 
Because covariance is bilinear, this gives us all the information 
we need in order to exactly compute $\cov(A_i, A_j)$.  
In particular, by setting $W_i=X_i-Y_i$, we have:
\begin{align*}
\cov(A_i,A_j) &= \cov(W_i^2-X_i-Y_i,W_j^2-X_j-Y_j)\\
&=\cov(X_i,X_j)+\cov(Y_i,Y_j)+ \cov(X_i,Y_j)+ \cov(X_j,Y_i) - \cov(W_i^2,X_j)\\
&\;\;\;\;- \cov(W_i^2,Y_j) - \cov(W_j^2,X_i)- \cov(W_j^2,Y_i)+\cov(W_i^2,W_j^2) \;.
\end{align*}
For the summands we have:
\begin{itemize}
\item[(a)]
\begin{align*}
\cov(W_i^2,X_j)&=\cov((X_i-Y_i)^2,X_j)=\cov(X_i^2,X_j)-2\cov(X_iY_i,X_j)\\
&=-mp_ip_j-2m(m-1)p_i^2p_j+2m^2p_ip_jq_i\\
&=-mp_ip_j(1-2p_i)+2m^2p_ip_j(q_i-p_i) \;.
\end{align*}

\item[(b)] \begin{align*}
\cov(W_i^2,Y_j)=-mq_iq_j(1-2q_i)+2m^2q_iq_j(p_i-q_i) \;.
\end{align*}

\item[(c)] \begin{align*}
\cov(W_j^2,X_i)=-mp_ip_j(1-2p_j)+2m^2p_ip_j(q_j-p_j) \;.
\end{align*}

\item[(d)] \begin{align*}
\cov(W_j^2,Y_i)=-mq_iq_j(1-2q_j)+2m^2q_iq_j(p_j-q_j) \;.
\end{align*}

\item[(e)] \begin{align*}
\cov(W_i^2,W_j^2)=&\cov(X_i^2,X_j^2)+\cov(Y_i^2,Y_j^2)+4\cov(X_iY_i,X_jY_j)\\
&-2\cov(X_i^2,X_jY_j)-2\cov(X_j^2,X_iY_i)-2\cov(Y_i^2,X_jY_j)-2\cov(Y_j^2,X_iY_i)\\
=&\cov(X_i^2,X_j^2)+\cov(Y_i^2,Y_j^2)+4\cov(X_iY_i,X_jY_j)\\
&-2\cov(X_i^2,X_j)\mathbb{E}[Y_j]-2\cov(X_j^2,X_i)\mathbb{E}[Y_i]\\
&-2\cov(Y_i^2,Y_j)\mathbb{E}[X_j]-2\cov(Y_j^2,Y_i)\mathbb{E}[X_i]\\
=&-mp_ip_j -2m(m-1)(p_i^2 p_j + p_i p_j^2) -2m(m-1)(2m-3) p_i^2 p_j^2 \\
&-mq_iq_j -2m(m-1)(q_i^2 q_j + q_i q_j^2) -2m(m-1)(2m-3) q_i^2 q_j^2\\
&+4(m^2-2m^3)p_i p_j q_i q_j +2m^2 p_i p_jq_j +4m^2(m-1) p_i^2 p_jq_j\\
&+2m^2 p_i p_jq_i +4m^2(m-1) p_j^2 p_iq_i+2m^2 q_i q_jp_j +4m^2(m-1) q_i^2 q_jp_j\\
&+2m^2 q_i q_jp_i +4m^2(m-1) q_j^2 q_ip_i \;.
\end{align*}
\end{itemize}

By substituting, we get:
\begin{align*}
\cov(A_i,A_j) &=\cov(X_i,X_j)+\cov(Y_i,Y_j) - \cov(W_i^2,X_j)\\
&\;\;\;\;- \cov(W_i^2,Y_j) - \cov(W_j^2,X_i)- \cov(W_j^2,Y_i)+\cov(W_i^2,W_j^2)\\
=&-m(p_ip_j+q_iq_j)\\
&+mp_ip_j(1-2p_i)-2m^2p_ip_j(q_i-p_i)+mq_iq_j(1-2q_i)-2m^2q_iq_j(p_i-q_i)\\
&+mp_ip_j(1-2p_j)-2m^2p_ip_j(q_j-p_j)+mq_iq_j(1-2q_j)-2m^2q_iq_j(p_j-q_j)\\
&+\cov(W_i^2,W_j^2)\\
=&-2m^2[p_ip_j(q_i+q_j)+q_iq_j(p_i+p_j)]+2m^2[p_ip_j(q_i+q_j)+q_iq_j(p_i+p_j)]\\
&-2m(m-1)(2m-3) p_i^2 p_j^2-2m(m-1)(2m-3) q_i^2 q_j^2\\
&+4(m^2-2m^3)p_i p_j q_i q_j+4m^2(m-1)(p_iq_j+p_jq_i)(p_ip_j+q_iq_j)\\
=&-6m(p_i^2p_j^2+q_i^2q_j^2)\\
&+m^2[10(p_i^2p_j^2+q_i^2q_j^2)+4p_i p_j q_i q_j-4(p_iq_j+p_jq_i)(p_ip_j+q_iq_j)]\\
&-m^3[4(p_i^2p_j^2+q_i^2q_j^2)+8p_i p_j q_i q_j-4(p_iq_j+p_jq_i)(p_ip_j+q_iq_j)]\\
=&-6m(p_i^2p_j^2+q_i^2q_j^2)\\
&+2m^2[5(p_i^2p_j^2+q_i^2q_j^2)+2p_i p_j q_i q_j-2(p_iq_j+p_jq_i)(p_ip_j+q_iq_j)]\\
&-4m^3[(p_ip_j+q_iq_j)^2-(p_iq_j+p_jq_i)(p_ip_j+q_iq_j)] \;.
\end{align*}


In summary,
\begin{align*}
\cov(A_i,A_j) = &-6 m (p_i^2 p_j^2 + q_i^2 q_j^2) \\
&+2m^2 [(5p_i^2 p_j^2 + 5q_i^2 q_j^2) -6p_i p_j q_i q_j  -2p_iq_i(p_j - q_j)^2 - 2p_jq_j(p_i - q_i)^2] \\
&-4 m^3 (p_i - q_i) (p_j - q_j) (p_i p_j + q_i q_j) \;.
\end{align*}

The total contribution of the covariances to the variance for all $i \neq j$ is $\sum_{i \neq j} \cov(A_i, A_j)$. 
We consider the coefficients on each of the powers of $m$ separately. We have:
\begin{align*}
[m^3]\sum_{i \neq j} \cov(A_i, A_j) &= - 4\sum_{i \neq j} (p_i-q_i)(p_j-q_j)(p_ip_j + q_iq_j) \\
&= 4\sum_i (p_i-q_i)^2 (p_i^2+q_i^2) - 4\sum_{i, j} (p_i-q_i)(p_j-q_j)(p_ip_j + q_iq_j) \\
&\leq 4\sum_i (p_i-q_i)^2 (p_i+q_i) - 4\sum_{i, j} (p_i-q_i)(p_j-q_j)(p_ip_j + q_iq_j) \\
&= 4\sum_i (p_i-q_i)^2 (p_i+q_i) -4(p-q)^\intercal (pp^\intercal + qq^\intercal) (p-q) \\
&\leq 4\sum_i (p_i-q_i)^2 (p_i+q_i).
\end{align*}

Also, $[m]\sum_{i \neq j} \cov(A_i, A_j) \leq 0$.

Finally, we have
\begin{align*}
[m^2]\sum_{i \neq j} \cov(A_i,A_j) &= 2 \sum_{i \neq j} [(5p_i^2 p_j^2 + 5q_i^2 q_j^2) -6p_i p_j q_i q_j  -2p_iq_i(p_j - q_j)^2 - 2p_jq_j(p_i - q_i)^2] \\
&\leq 10\sum_{i \neq j} (p_i^2 p_j^2 + q_i^2 q_j^2) \\
&\leq 10\sum_{i, j} (p_i^2 p_j^2 + q_i^2 q_j^2) \\
&= 10[p^\intercal (p p^\intercal) p + q^\intercal (q q^\intercal) q] \\
&=10[(p^\intercal p) (p^\intercal p) + (q^\intercal q) (q^\intercal q)] \\
&= 10\|p\|_2^4 + 10\|q\|_2^4 \\
&\leq 20 b^2  \leq 20b.
\end{align*}

\subsection{Completing the Proof}

\begin{align*}
\var[A]=&\sum_{i=1}^n \var[A_i]+\sum_{i\not= j} \cov(A_i,A_j) \\
\leq & \sum_{i=1}^n 80m^2 \left(\frac{p_i+ q_i}{2}\right)^2 + 4m^3 (p_i-q_i)^2 (p_i+q_i)\\
&+20m^2b + 4m^3\sum_i (p_i-q_i)^2 (p_i+q_i)\\
\leq &100m^2b+8  m^3\sum_i (p_i-q_i) (p_i^2-q_i^2) \;.
\end{align*}
\qed

\bibliographystyle{alpha}


\bibliography{allrefs}

\newcommand{\etalchar}[1]{$^{#1}$}
\begin{thebibliography}{CDGR16}

\bibitem[ADJ{\etalchar{+}}12]{Orlitsky:colt12}
J.~Acharya, H.~Das, A.~Jafarpour, A.~Orlitsky, S.~Pan, and A.~Suresh.
\newblock Competitive classification and closeness testing.
\newblock In {\em COLT}, 2012.

\bibitem[BFF{\etalchar{+}}01]{BFFKRW:01}
T.~Batu, E.~Fischer, L.~Fortnow, R.~Kumar, R.~Rubinfeld, and P.~White.
\newblock Testing random variables for independence and identity.
\newblock In {\em Proc. 42nd IEEE Symposium on Foundations of Computer
  Science}, pages 442--451, 2001.

\bibitem[BFR{\etalchar{+}}00]{BFR+:00}
T.~Batu, L.~Fortnow, R.~Rubinfeld, W.~D. Smith, and P.~White.
\newblock Testing that distributions are close.
\newblock In {\em {IEEE} Symposium on Foundations of Computer Science}, pages
  259--269, 2000.

\bibitem[BFR{\etalchar{+}}13]{Batu13}
T.~Batu, L.~Fortnow, R.~Rubinfeld, W.~D. Smith, and P.~White.
\newblock Testing closeness of discrete distributions.
\newblock {\em J. ACM}, 60(1):4, 2013.

\bibitem[Can15]{Canonne15}
C.~L. Canonne.
\newblock A survey on distribution testing: Your data is big. but is it blue?
\newblock {\em Electronic Colloquium on Computational Complexity {(ECCC)}},
  22:63, 2015.

\bibitem[CDGR16]{CDGR16}
C.~L. Canonne, I.~Diakonikolas, T.~Gouleakis, and R.~Rubinfeld.
\newblock Testing shape restrictions of discrete distributions.
\newblock In {\em 33rd Symposium on Theoretical Aspects of Computer Science,
  {STACS}}, pages 25:1--25:14, 2016.

\bibitem[CDVV14]{CDVV14}
S.~Chan, I.~Diakonikolas, P.~Valiant, and G.~Valiant.
\newblock Optimal algorithms for testing closeness of discrete distributions.
\newblock In {\em SODA}, pages 1193--1203, 2014.

\bibitem[DDS{\etalchar{+}}13]{DDSVV13}
C.~Daskalakis, I.~Diakonikolas, R.~Servedio, G.~Valiant, and P.~Valiant.
\newblock Testing $k$-modal distributions: Optimal algorithms via reductions.
\newblock In {\em SODA}, pages 1833--1852, 2013.

\bibitem[DK16]{DK16}
I.~Diakonikolas and D.~M. Kane.
\newblock A new approach for testing properties of discrete distributions.
\newblock {\em CoRR}, abs/1601.05557, 2016.
\newblock In FOCS'16.

\bibitem[DKN15a]{DKN:15:FOCS}
I.~Diakonikolas, D.~M. Kane, and V.~Nikishkin.
\newblock Optimal algorithms and lower bounds for testing closeness of
  structured distributions.
\newblock In {\em 56th Annual {IEEE} Symposium on Foundations of Computer
  Science, {FOCS} 2015}, 2015.

\bibitem[DKN15b]{DKN:15}
I.~Diakonikolas, D.~M. Kane, and V.~Nikishkin.
\newblock {T}esting {I}dentity of {S}tructured {D}istributions.
\newblock In {\em Proceedings of the Twenty-Sixth Annual {ACM-SIAM} Symposium
  on Discrete Algorithms, {SODA} 2015}, 2015.

\bibitem[Gol16a]{Goldreich16}
O.~Goldreich.
\newblock The uniform distribution is complete with respect to testing identity
  to a fixed distribution.
\newblock {\em Electronic Colloquium on Computational Complexity {(ECCC)}},
  23:15, 2016.

\bibitem[Gol16b]{Goldreich16-notes}
O.~Goldreich.
\newblock {Lecture Notes on Property Testing of Distributions}.
\newblock Available at http://www.wisdom.weizmann.ac.il/~oded/PDF/pt-dist.pdf,
  March, 2016.

\bibitem[GR00]{GRdist:00}
O.~Goldreich and D.~Ron.
\newblock On testing expansion in bounded-degree graphs.
\newblock {\em Electronic Colloqium on Computational Complexity}, 7(20), 2000.

\bibitem[LR05]{lehmann2005testing}
E.~L. Lehmann and J.~P. Romano.
\newblock {\em Testing statistical hypotheses}.
\newblock Springer Texts in Statistics. Springer, 2005.

\bibitem[NP33]{NeymanP}
J.~Neyman and E.~S. Pearson.
\newblock On the problem of the most efficient tests of statistical hypotheses.
\newblock {\em Philosophical Transactions of the Royal Society of London.
  Series A, Containing Papers of a Mathematical or Physical Character},
  231(694-706):289--337, 1933.

\bibitem[Pan08]{Paninski:08}
L.~Paninski.
\newblock A coincidence-based test for uniformity given very sparsely-sampled
  discrete data.
\newblock {\em IEEE Transactions on Information Theory}, 54:4750--4755, 2008.

\bibitem[Rub12]{Rub12}
R.~Rubinfeld.
\newblock Taming big probability distributions.
\newblock {\em XRDS}, 19(1):24--28, 2012.

\bibitem[VV14]{VV14}
G.~Valiant and P.~Valiant.
\newblock An automatic inequality prover and instance optimal identity testing.
\newblock In {\em FOCS}, 2014.

\end{thebibliography}

\end{document}